\documentclass[runningheads,a4paper]{llncs}
\usepackage{url}
\usepackage{epsf}
\usepackage{graphicx}
\usepackage{amsfonts}
\usepackage{amssymb}
\usepackage{amsmath}
\usepackage{latexsym}
\usepackage{color}
\usepackage{multirow}
\usepackage{setspace, float}
\usepackage{refcount}
\usepackage{hyperref}
\usepackage{makeidx}  

\usepackage{thmtools,thm-restate}

\newcommand{\arr}[1]{\ensuremath{\protect\overrightarrow{#1}}}

\newcommand{\ABox}{
\raisebox{3pt}{\framebox[6pt]{\rule{6pt}{0pt}}}
}

\newcommand{\pp}{\xi}
\newcommand{\red}[1]{\textcolor{red}{#1}}
\newcommand{\ang}{\measuredangle}

\urldef{\mailmd}\path|mirela.damian@villanova.edu|
\urldef{\mailvv}\path|dvoicu@villanova.edu|

\begin{document}

\mainmatter  

\title{Spanning Properties of Theta-Theta Graphs\thanks{This work was supported by NSF grant CCF-1218814.}}
\titlerunning{Spanning Properties of Theta-Theta Graphs}

\author{Mirela Damian and Dumitru V. Voicu}
\authorrunning{M. Damian and D. V. Voicu}

\institute{Department of Computer Science \\ Villanova University, Villanova, PA 19085 \\
\mailmd\\
\mailvv}

\date{}

\maketitle

\begin{abstract}
We study the spanning properties of Theta-Theta graphs. Similar in spirit with the Yao-Yao graphs, Theta-Theta graphs partition the space around each vertex into a set of $k$ cones, for some fixed integer $k > 1$, and select at most one edge per cone. The difference is in the way edges are selected. Yao-Yao graphs select an edge of minimum length, whereas Theta-Theta graphs select an edge of minimum orthogonal projection onto the cone bisector. It has been established that the Yao-Yao graphs with parameter $k = 6k'$ have spanning ratio $11.67$, for $k' \ge 6$. In this paper we establish a first spanning ratio of $7.82$ for Theta-Theta graphs, for the same values of $k$. We also extend the class of Theta-Theta spanners with parameter $6k'$, and establish a spanning ratio of $16.76$ for $k' \ge 5$. We surmise that these stronger results are mainly due to a tighter analysis in this paper, rather than Theta-Theta being superior to Yao-Yao as a spanner. We also show that the spanning ratio of Theta-Theta graphs decreases to $4.64$ as $k'$ increases to $8$.  These are the first results on the spanning properties of Theta-Theta graphs. 
\end{abstract}

\section{Introduction}
Let $S$ be a set of $n$ points in the plane, and let $G$ be an undirected plane graph with vertex set $S$. The \emph{length} of a path in $G$ is the sum of the Euclidean lengths of its constituent edges. The distance in $G$ between any two points $a, b \in S$ is the length of a shortest path between $a$ and $b$. We say that $G$ is a \emph{spanner} if it preserves distances between each pair of points in $S$, up to a given factor. Specifically, for a fixed integer $t \ge 1$, we say that $G$ is a $t$-\emph{spanner} if any two points $a, b \in S$ at distance $|ab|$ in the plane are at distance at most $t\cdot |ab|$ in $G$. The smallest integer $t$ for which this property holds is called the \emph{spanning ratio} of $G$. Clearly there is a tradeoff between the spanning ratio and the sparsity of $G$: the smaller the spanning ratio, the denser the spanner and the better the approximation of the original distances. 

One way to control the tradeoff between the spanning ratio and the sparsity of the spanner is to partition the space around each point into equiangular cones of angle $\theta=2\pi/k$, for some integer $k \ge 1$, and connect each point to a ``nearest'' point in each cone. Intuitively, this construction promises a short detour between any two points $a, b \in S$, by following the edge from $a$ aiming in the direction of $b$ (the one lying in the cone with apex $a$ containing $b$). 
%
The definition of a ``nearest'' point comes in two flavors, in the context of \emph{Yao} graphs~\cite{Yao82} and \emph{Theta}-graphs (or $\Theta$-graphs)~\cite{Clark87,Keil88}. For Yao graphs, the ``nearest'' point is simply the point that minimizes the $L_2$-distance, whereas for Theta graphs, the ``nearest'' point in a cone $C$ is the point whose orthogonal projection onto the bisector of $C$ minimizes the $L_2$-distance.  Both Yao and Theta graphs are parameterized by a positive integer $k \ge 1$, which controls the cone angle $\theta = 2\pi/k$. 
In the following we will refer to the Yao graph as $Y_k$ and Theta graphs as $\Theta_k$, for a fixed $k \ge 1$.  Both $Y_k$ and $\Theta_k$ are known to be efficient spanners, for $k \ge 6$. The spanning ratios of these graphs 
are summarized in~\autoref{tab:spanningratios}.
%

\begin{table}
\begin{center}
\begin{tabular} {|c|c|c|c|c|c|}
\hline
\multirow{2}{*}{Parameter $k$} & \multicolumn{5}{|c|}{Spanning Ratio}  \\
\cline{2-6}
& $Y_k$ & \multicolumn{2}{|c|}{$\Theta_k$} & $YY_k$ & $\Theta\Theta_k$ \\
\hline \hline
\small{$< 4$} & \multicolumn{5}{|c|}{\small{$\infty$~\cite{MollaThesis09}}} \\
\hline
\small{$4$} & \small{$696.1$~\cite{BDD+10}} & \multicolumn{2}{|c|}{\small{237~\cite{BarbaBD13}}} & \multicolumn{2}{|c|}{\small{$\infty$~\cite{DMP09}}} \\
\hline
\small{$5$} & \small{$3.74$~\cite{BarbaBD14}} & \multicolumn{2}{|c|}{\small{$9.96$~\cite{BoseMRS14}}} &  
\small{OPEN} & \small{$\infty$~\cite{KX13}} \\
\hline
\small{$6$} & \small{5.8~\cite{BarbaBD14}}  & \multicolumn{2}{|c|}{\small{$2$~\cite{BGH+10}}} & \small{$\infty$~\cite{MollaThesis09}} & \small{OPEN} \\
\hline
\small{$k > 6$} & \multirow{3}{*}{$\frac{1}{1-2\sin(\theta/2)}$~\small{\cite{BDDx10}}} & \multicolumn{2}{|c|}{$\frac{1}{1-\sin(\theta/2)}$~\small{\cite{RS91}} }
& \raisebox{-0.3em}{\small{$11.67$~for}} & \raisebox{-0.3em}{\red{\small{$16.76$~for}}} \\
\cline{1-1} \cline{3-4}
\small{$4k+2$} & & \small{$1+2\sin(\theta/2)$} & \multirow{3}{*}{\small{\cite{BoseRV13,BoseCM+14}}} &  \small{$k=6k'$~and} & \red{\small{$k=6k'$~and}} \\
\cline{1-1} \cline{3-3}
\small{$4k+4$} & & $1+\frac{2\sin(\theta/2)}{\cos(\theta/2)-\sin(\theta/2)}$ & &  \raisebox{0.3em}{\small{$k' \ge 6$}} &  \raisebox{0.3em}{\red{\small{$k' \ge 5$}}}  \\
\cline{1-3}
\small{$4k+3, 4k+5$} & $\frac{1}{1-2\sin(3\theta/8)}$~\small{\cite{BarbaBD14}} & $\frac{\cos(\theta/4)}{\cos(\theta/2)-\sin(3\theta/4})$ & &  \raisebox{0.4em}{\small{\cite{DB13}}} & \raisebox{0.4em}{\small{\red{[HERE]}}} \\
\hline
\end{tabular}
\end{center}
\caption{Spanning ratios of Yao and Theta graphs for various $\theta=2\pi/k$ values.}
\label{tab:spanningratios}
\end{table}

\vspace{-2em}
Interest in Yao and Theta graphs has increased with the advancement of wireless ad hoc networks and the need for efficient communication (see~\cite{RodittyMU08,Chris08,KanjPGe09,CowenLW00} 
and the references therein). Designing routing algorithms for wireless ad hoc networks is an extremely difficult 
task and research in this area is still in progress. The overlay communication graph formed by the wireless links should be a spanner to ensure fast delivery of information, and should also have low degree to ensure a low maintenance cost and reduced MAC-level contention and interference~\cite{HamR06}. We observe that both Yao and Theta graphs obey the first requirement (as detailed in   \autoref{tab:spanningratios}), but fail to satisfy the second requirement. One simple example consists of $n-1$ points equally distributed around a circle centered at an n$^{th}$ point $p$. Then, for $k \ge 6$, both $\Theta_k$ and $Y_k$ will have an edge directed from each of the $n-1$ points towards $p$, because $p$ is ``nearest'' in one of their cones. So each of $\Theta_k$ and $Y_k$ has out-degree $k$, but in-degree $n-1$. To reduce the in-degree, alternate spanner structures based on Yao and Theta graphs have been proposed, such as Yao-Yao~\cite{WL03}, Sink~\cite{LiWanWang01,AryaYY95}, Stable Roommates~\cite{BoseCCCKL13}, and Ordered-Yao~\cite{Song04}. 

The \emph{Yao-Yao} graph with integer parameter $k \ge 1$, denoted $YY_k$, is a subgraph of $Y_k$ obtained by applying a second Yao step to the set of incoming edges in each cone.  More precisely, for each point $p$ and each cone with apex $a$ containing two or more incoming edges, $YY_k$ retains only a shortest incoming edge and discards the rest. Ties are broken arbitrarily.  This construction guarantees a degree of at most $2k$ at each node in $YY_k$ (one incoming and one outgoing edge per cone), however the spanning property of $YY_k$ is still under investigation.  The only existing result shows that $YY_{6k'}$, for $k' \ge 6$, is a spanner with spanning ratio $11.67$. For $k' \ge 8$,  the spanning ratio of $YY_{6k'}$ drops to $4.75$~\cite{DB13}. 

\emph{Sink} spanners~\cite{LiWanWang01,AryaYY95} transform bounded outdegree spanners, such as $Y_k$ and $\Theta_k$, into bounded degree spanners, by replacing  each directed star consisting of all links directed into a point $p$ and lying in a cone with apex $p$, 
by a tree of bounded degree with ``sink'' $p$. The result is a spanner with degree at most $k(k+2)$ and spanning ratio $1/(1-2\sin(\theta/2))^2$. 

The \emph{Stable Roommates} spanner introduced in~\cite{BoseCCCKL13} has degree 
at most $k$ and spanning ratio matching the spanning ratio of $Y_k$, so this spanner combines both qualities -- low spanning ratio and low degree -- of the Yao and Yao-Yao graphs, respectively. 
The only drawback of this approach is that it processes pairs of points in non-decreasing order by their distances, making it unsuitable for a fast local implementation. (The authors present a distributed implementation that requires $O(n)$ rounds of communication.) 

The \emph{ordered} Theta approach~\cite{BoseGM04} reduces the potentially linear degree of the Theta graph to a logarithmic degree.  
Similar to the stable roommates approach, the ordered Theta approach imposes a particular ordering on the input points. 
The authors show that careful orderings can produce graphs with spanning ratio $1/(\cos\theta-\sin\theta)$ and degree $O(k \log n)$.


Similar in spirit with the Yao-Yao graph, in this paper we introduce the \emph{Theta-Theta} graph $\Theta\Theta_k$, parameterized by integer $k \ge 1$, and study the spanning properties of this graph.  The graph $\Theta\Theta_k$ is obtained by applying a filtering step to the edges of $\Theta_k$ as follows. For each point $p$ and each cone $C$ with apex $p$, we consider all edges in $C$ directed into $p$, and maintain only a ``shortest'' edge while discarding the rest. Recall that in the context of Theta graphs, a ``shortest'' edge minimizes the length of its projection on the cone bisector.  Ties are arbitrarily broken.

Our main result shows that $\Theta\Theta_{6k'}$ is a spanner, for any $k' \ge 5$. This result relies on a result by Bonichon et al.~\cite{BGH+10}, who prove that $\Theta_6$ is a $2$-spanner. Our main contribution is showing that  $\Theta\Theta_{6k'}$ contains a short path between the endpoints of each edge in $\Theta_6$. More precisely, we show that for each edge $ab \in \Theta_6$, there is a path between $a$ and $b$ in $\Theta\Theta_{6k'}$ no longer than $8.38|ab|$, for $k' \ge 5$. This, combined with the fact that $\Theta_6$ is a $2$-spanner, yields an upper bound of $16.76$ on the spanning ratio of $\Theta\Theta_{6k'}$. 
A similar approach has been used in~\cite{DB13} to establish that $YY_{6k'}$ has spanning ratio $11.67$, for $k' \ge 6$.  
We observe that the spanning ratio of $\Theta\Theta_{6k'}$ decreases to $7.82$, $5.63$ and $4.64$ as $k'$ increases to $6$, $7$, and above $8$, respectively. 
The spanning ratios established in this paper for $\Theta\Theta_{6k'}$ are stronger than the ones obtained in~\cite{DB13} for $YY_{6k'}$, for the same parameter values $k' \ge 6$. We surmise that this is mainly due to the tighter analysis in this paper, rather than $\Theta\Theta_{6k'}$ being superior to $YY_{6k'}$ as a spanner.




\subsection{Definitions}
\label{sec:defs}
Throughout the paper, $S$ will refer to a fixed set of $n$ points in the plane. The directed Yao graph $Y_k$ with integer parameter $k \ge 1$ on $S$ is constructed as follows. For each point $a \in S$, starting with the direction of the positive $x$-axis, extend $k$ equally spaced rays $r_1, r_2, \ldots, r_k$ originating at $a$, in counterclockwise order (see \autoref{fig:defs}a for $k = 6$). These rays divide the plane into $k$ cones, denoted by $C_{\{k,1\}}(a), C_{\{k,2\}}(a), \ldots, C_{\{k,k\}}(a)$, each of angle $\theta = 2\pi/k$. To avoid overlapping boundaries, we assume that each cone is half-open and half-closed, meaning that $C_{\{k,i\}}(a)$ includes $r_i$ but excludes $r_{i+1}$ (here $r_{k+1} \equiv r_1$ wraps around).  
In each cone of $a$, draw a directed edge from $a$ to its ``closest'' point $b$ in that cone (the one that minimizes the $L_2$-distance $|ab|$). Ties are broken arbitrarily. These directed edges collectively form the edge set for the directed Yao graph. The undirected Yao graph (or simply Yao graph) on $S$ is obtained by simply ignoring the directions of these edges. The Theta graph $\Theta_k$ is defined in a similar way, with the only difference being in the definition of ``closest'': in each cone $C$ with apex $a$, draw a directed edge from $a$ to the point $b$ that minimizes the distance between $a$ and the orthogonal projection of $b$ on the bisector of the cone.  
%
\begin{figure}[pht]
\centering
\begin{tabular}{c@{\hspace{0.1\linewidth}}c}
\includegraphics[width=0.35\linewidth]{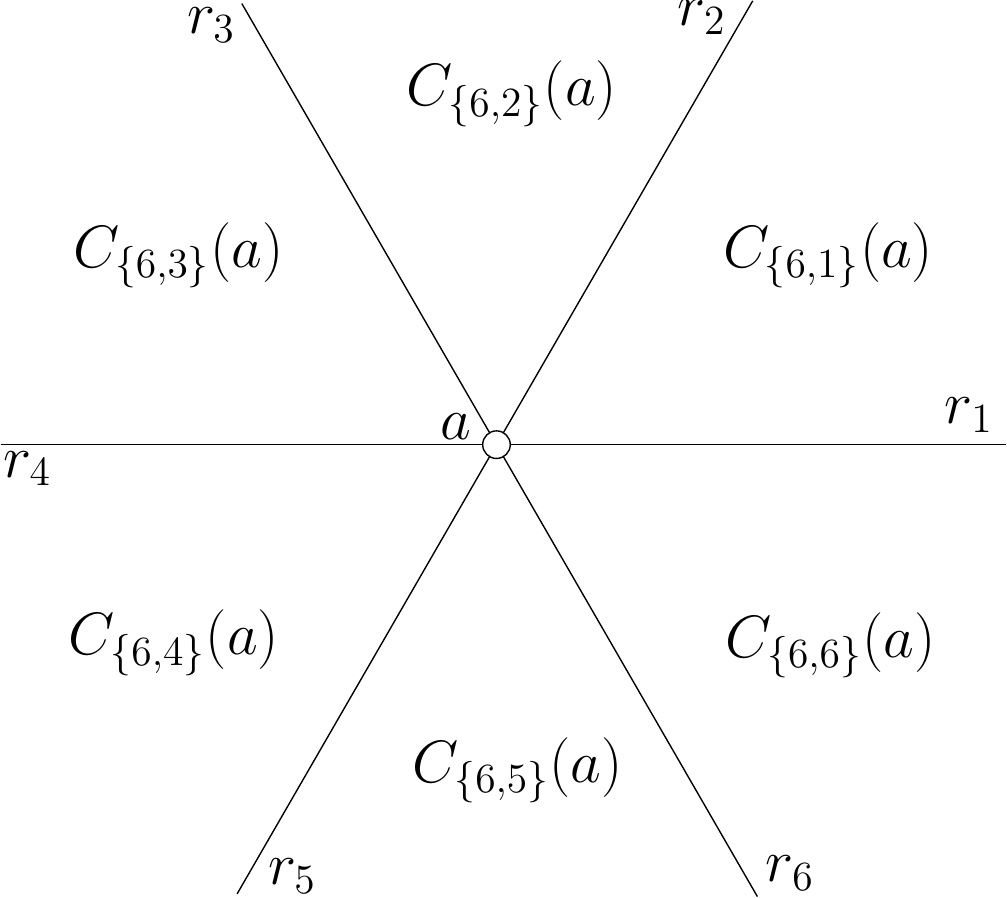} & 
\raisebox{2em}{\includegraphics[width=0.42\linewidth]{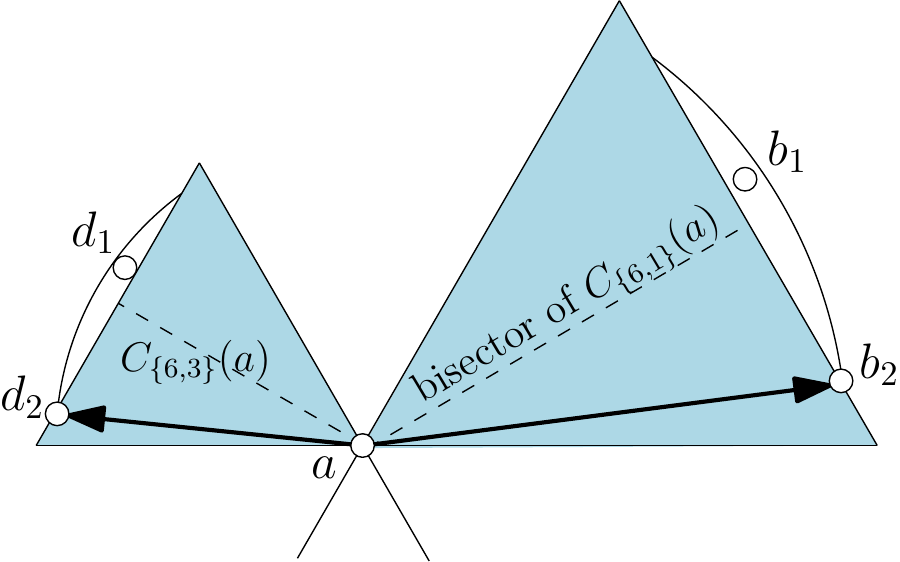}} \\
(a) & (b) 
\end{tabular}
\caption{Definitions (a) Rays defining the cones at point $a$ (b) Theta edges $ab_2$, $ad_2$.}
\label{fig:defs}
\end{figure}
%
For example, looking at the cone $C_{\{6,1\}}(a)$ in~\autoref{fig:defs}b, notice that $b_1$ minimizes the $L_2$-distance to $a$, whereas $b_2$ minimizes the $L_2$-distance between its projection onto the cone bisector and $a$. Consequently, $\arr{a b_1}$ will be added to $Y_6$, and $\arr{a b_2}$ to $\Theta_6$.  Similarly, $\arr{ad_1} \in C_{\{6,3\}}(a)$ will be added to $Y_6$, and $\arr{ad_2}$ to $\Theta_6$.
\autoref{fig:yaothetaex}a shows the Yao graph $Y_6$ for the point set depicted in \autoref{fig:defs}b, and  
\autoref{fig:yaothetaex}c shows the Theta graph $\Theta_6$ for the same point set.

The Yao-Yao graph $YY_k \subseteq Y_k$ is obtained from $Y_k$ by applying a reverse Yao step to the set of incoming Yao edges in $Y_k$. That is, for each node $a$  and each cone with apex $a$ containing two or more incoming edges, $YY_k$ retains a shortest incoming edge and discards the rest. Ties are broken arbitrarily.  The Theta-Theta graph $\Theta\Theta_k \subseteq \Theta_k$ is obtained from $\Theta_k$ in a similar way, with the only difference being in the requirement that a ``shortest'' incoming edge 
in a  cone minimizes the length of its projection onto the cone bisector.
\autoref{fig:yaothetaex}b shows the graph $YY_6$ derived from the graph $Y_6$ depicted in  \autoref{fig:yaothetaex}a, and  
\autoref{fig:yaothetaex}d shows the graph $\Theta\Theta_6$ derived from the graph $\Theta_6$ depicted in  \autoref{fig:yaothetaex}c. 

When the choice of a particular cone is either irrelevant or is clear from the context, we ignore the cone subscript and use $C_k(a)$ to denote any of the cones $C_{\{k,1\}}(a), C_{\{k,2\}}(a), \ldots C_{\{k,k\}}(a)$.  For any two points $a, b \in S$, let $C_k(a,b)$ denote the cone with apex $a$ that contains $b$.  Let $\triangle_k(a, b)$ be the canonical 
triangle with two of its sides along the rays bounding $C_{k}(a,b)$, and the third side orthogonal to the bisector of 
$C_{k}(a,b)$ and passing through $b$. For example, shaded in~\autoref{fig:defs}b are the canonical triangles $\triangle_6(a, b_2)$ and $\triangle_6(a, d_2)$. 

\begin{figure}[htbp]
\centering
\begin{tabular}{c@{\hspace{0.1\linewidth}}c}
\includegraphics[width=0.4\linewidth]{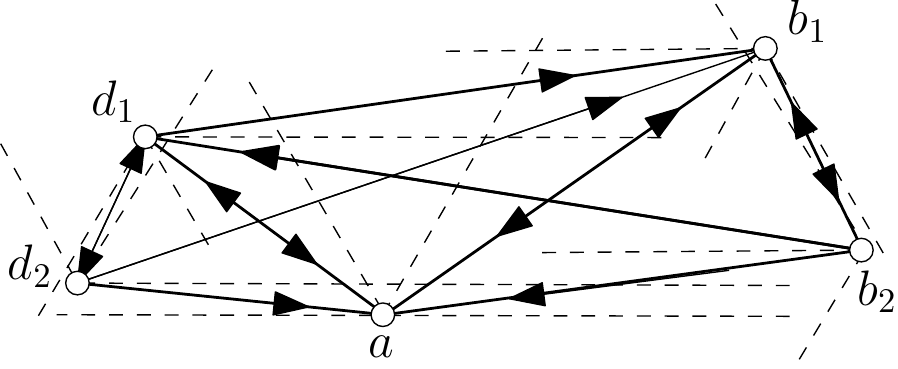} & 
\includegraphics[width=0.4\linewidth]{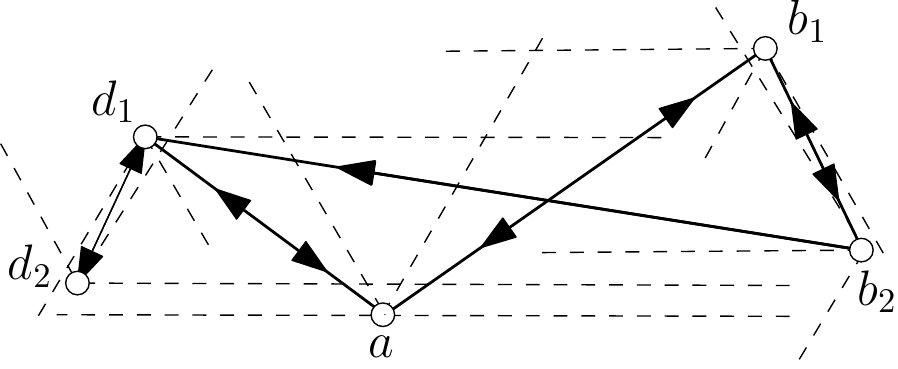} \\
(a) & (b) \\
\includegraphics[width=0.4\linewidth]{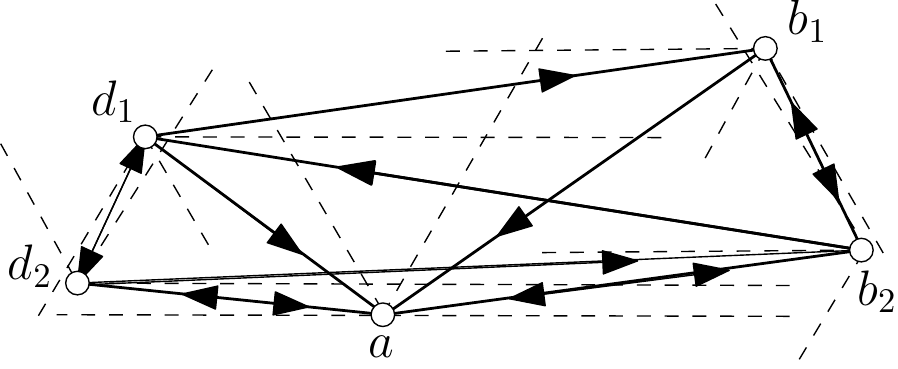} & 
\includegraphics[width=0.4\linewidth]{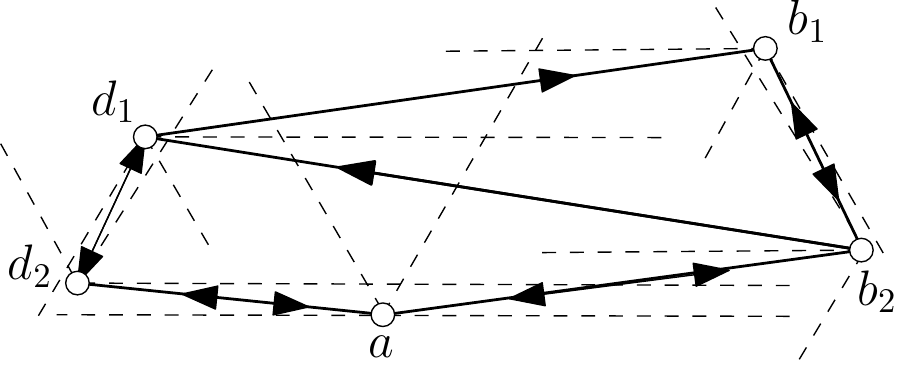} \\
(c) & (d) 
\end{tabular}
\caption{Graph examples (a) $Y_6$ (b) $YY_6$ (c) $\Theta_6$ (d) $\Theta\Theta_6$.}
\label{fig:yaothetaex}
\end{figure}

For any pair of vertices $a$ and $b$ in an undirected graph $G$, let $\pp_G(a, b)$ denote a \emph{shortest} path 
in $G$ between $a$ and $b$. 
For example, $\pp_{\Theta_6}(a, b)$ refers to a shortest path in $\Theta_6$ from $a$ to $b$. 

Our main goal is to establish a short path in $\Theta\Theta_k$ between the endpoints of each edge in $\Theta_6$. Our arguments will rely on the assumption that, for each point $a \in S$, each cone $C_k(a)$ is entirely contained in $C_6(a)$, hence $k = 6k'$. 
Throughout the rest of the paper, will will work with a quadruple of distinct points $a, b, b', a' \in S$ in the following configuration: $\arr{ab}$ is an arbitrary edge in $\Theta_6$; $\arr{ab'}$ is the edge in $\Theta_k$ that lies in the cone $C_k(a, b) \subset C_6(a,b)$; and $\arr{a'b'}$ is the edge in $\Theta\Theta_k$ that lies in the cone $C_k(b',a) \subset C_6(b',a)$. We will refer to this configuration as a \emph{canonical $\Theta$-configuration}, to avoid repeating these definitions in different contexts.
For a snapshot of a canonical $\Theta$-configuration, see ahead to~\autoref{fig:abba}a. 
We will further assume, without loss of generality, that in a canonical $\Theta$-configuration $\arr{ab}$ lies in $C_{\{6,1\}}(a)$, and the bisector of $C_k(a, b)$ lies below, or aligns with, the bisector of $\triangle_6(a,b)$. Any other configuration is equivalent to this canonical $\Theta$-configuration under rotational and/or reflectional symmetry.  

\section{Preliminaries}
\label{sec:basic}
In this section we present a few isolated lemmas that will be used in our main proof from~\autoref{sec:main}. 
For the sake of clarity and continuity in the flow of our exposition, we defer the proofs of most of these lemmas to the appendix. 
We encourage the reader to skip ahead to~\autoref{sec:main}, and refer back to these lemmas from the context of~\autoref{thm:maintheta}, where their role will become evident.
%
We begin this section with the statement of an existing result. 

\begin{theorem}{\emph{\cite{BGH+10}}}
For any pair of points $a, b\in S$, there is a path in $\Theta_6$ whose total length is bounded above by $2|ab|$.
\label{thm:theta6}
\end{theorem}
The key ingredient in the result of~\autoref{thm:theta6} is a specific subgraph of $\Theta_6$, called \emph{half-}$\Theta_6$. This graph preserves half of the edges in $\Theta_6$, those belonging to non-consecutive cones. 
Bonichon et al.~\cite{BGH+10} show that half-$\Theta_6$ is a
\emph{triangular-distance}\footnote{The \emph{triangular distance} from a point $a$ to a point $b$
is the side length of the smallest equilateral triangle centered at $a$ that touches $b$ and has one horizontal side.} 
Delaunay triangulation, computed as the dual of the Voronoi diagram based on
the triangular distance function. Combined with Chew's proof that any triangular-distance Delaunay triangulation is a $2$-spanner~\cite{Chew89}, this result settles \autoref{thm:theta6}. 
The structure of $\Theta_6$, viewed as the union of two planar $2$-spanners,  has been used in establishing spanning properties of other graphs as well~\cite{Bon2+10,jDR12,DB13}.

\begin{figure}[htpb]
\centering
\includegraphics[width=0.4\linewidth]{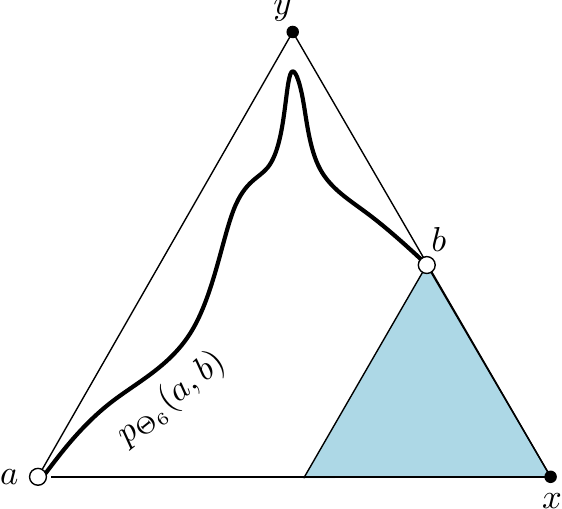} 
\caption{\autoref{lem:thetapath}: $|\pp_{\Theta_6}(a, b)| \le |ay| + |by|$.}
\label{fig:trapezoid}
\end{figure}

\begin{figure}[htpb]
\centering
\begin{tabular}{c@{\hspace{0.05\linewidth}}c}
\includegraphics[width=0.4\linewidth]{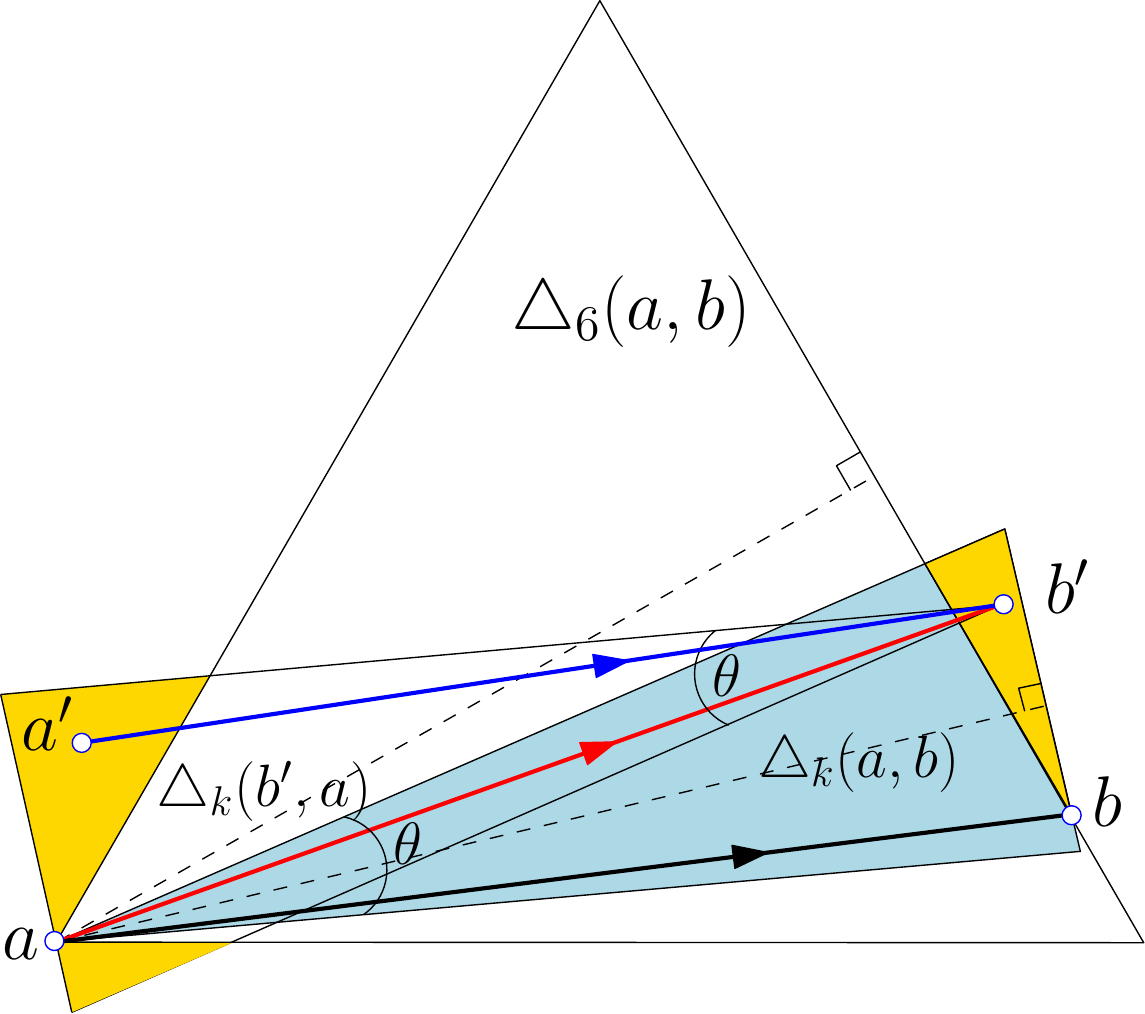} & 
\includegraphics[width=0.43\linewidth]{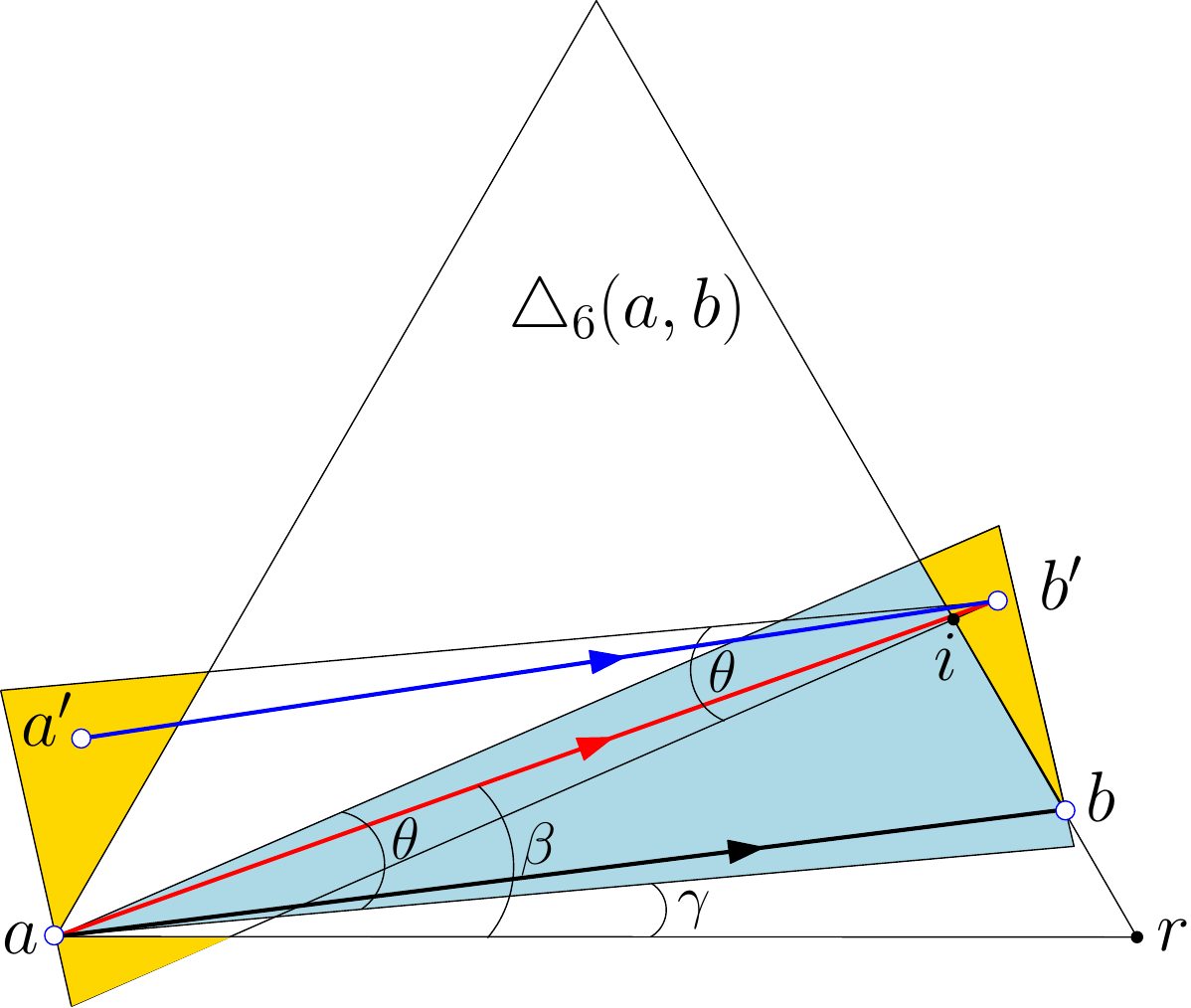} \\
(a) & (b) 
\end{tabular}
\caption{(a) Canonical $\Theta$-configuration: $ab \in \Theta_6$, $ab' \in \Theta_k$ and $a'b' \in \Theta\Theta_k$ (b) Bounding $|ab'|$, $|a'b'|$ and $|bb'|$.}
\label{fig:abba}
\end{figure}

\noindent
Before stating our preliminary results, we define the term $T(\alpha$) parameterized by angle $\alpha \in [0, \pi/3]$ as 
\begin{equation}
T(\alpha) = \frac{\sin(\pi/3-\alpha)-\sin\alpha}{\sin(\pi/3)} \le 1
\label{eq:T}
\end{equation}
This term will occur frequently in our analysis, and this definition will come in handy. The upper bound of $1$ follows from the fact that $T(\alpha$) decreases as $\alpha$ increases, therefore $T(\alpha) \le T(0) = 1$.  
The following lemma plays a central role in the proofs of Lemmas~\ref{lem:paa1} and~\ref{lem:paasecond}. 
\begin{lemma}{\emph{\cite{DB13}}}
Let $a, b \in S$ and let $x$ and $y$ be the other two vertices of $\triangle_6(a, b)$. 
If $\triangle_6(b,x)$ is empty of points in $S$, then $|\pp_{\Theta_6}(a, b)| \le |ay| + |by|$. Moreover, each edge of $\pp_{\Theta_6}(a, b)$ is no longer than $|ay|$. \emph{[Refer to~\autoref{fig:trapezoid}.]}
\label{lem:thetapath}
\end{lemma}
Note that~\autoref{lem:thetapath} does not specify which of the two sides $ax$ and $ay$ lies clockwise from $\triangle_6(a,b)$, so the lemma applies in both situations. 
%
The following lemma establishes fundamental relationships on the distances between points in a canonical $\Theta$-configuration. 

\begin{restatable}{lemma}{abbalemma}
\label{lem:abba}
Let $a, b, b', a' \in S$ be points in a canonical $\Theta$-configuration. Then each of $|ab'|$ and $|a'b'|$ is no longer than $|ab|/\cos(\theta/2)$. In addition, if $\beta$ and $\gamma$ are the angles formed by the horizontal through $a$ with $ab'$ and the lower ray of $C_k(a, b)$, respectively, and if $\beta \le \pi/6$, then 
\begin{eqnarray*}
|ab'| \ge |ab|\frac{\sin(\pi/3+\gamma)}{\sin(\pi/3+\beta)} 
\end{eqnarray*}
\emph{[Refer to~\autoref{fig:abba}b.]} 
\end{restatable}

\noindent
Lemmas~\ref{lem:paa1} through~\ref{lem:paa5} isolate specific situations that will arise in the analysis of our main result. 
We state 
them independently in this section.


\begin{figure}[htbp]
\centering
\begin{tabular}{c@{\hspace{0.05\linewidth}}c}
\includegraphics[width=0.45\linewidth]{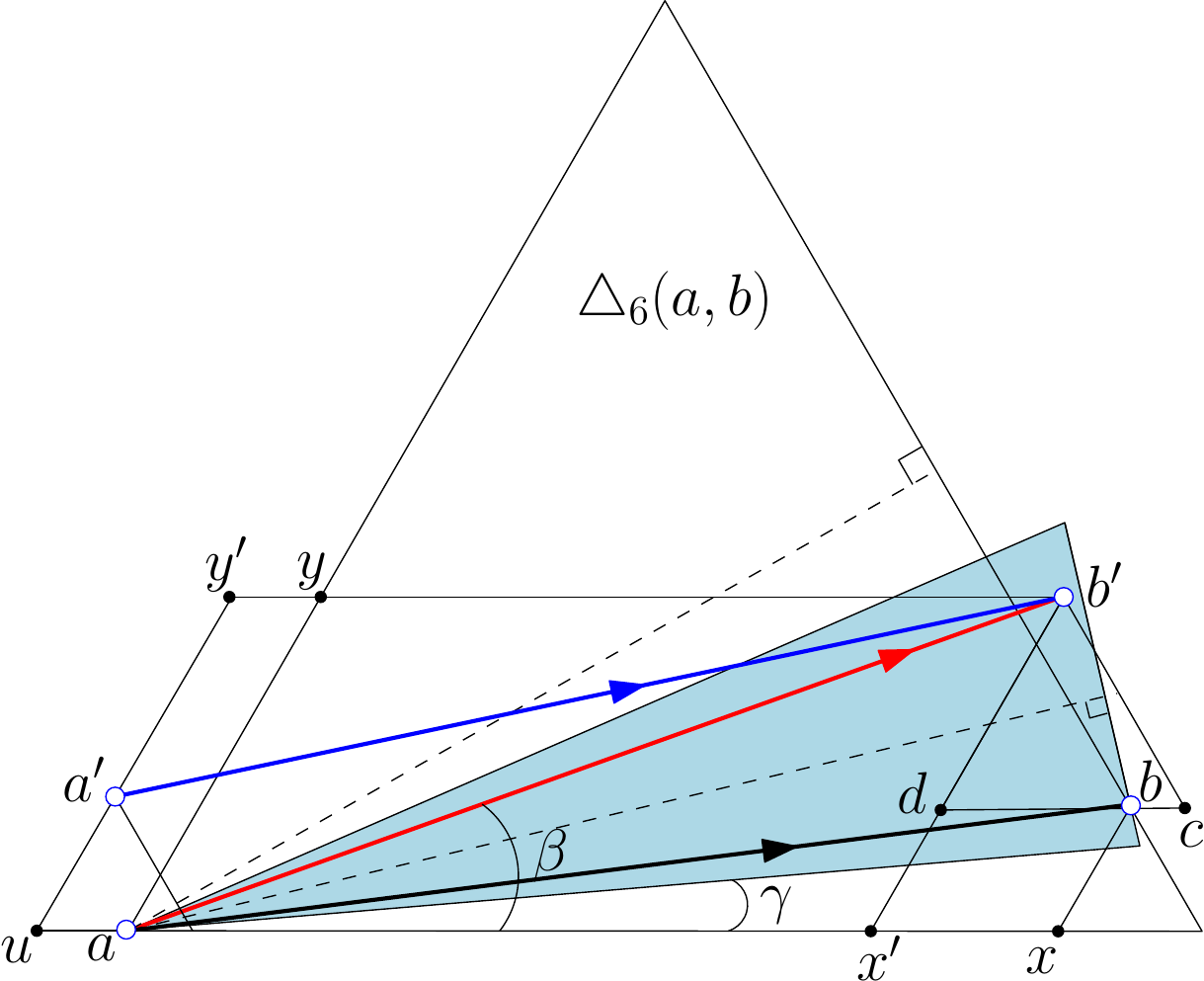} &
\includegraphics[width=0.45\linewidth]{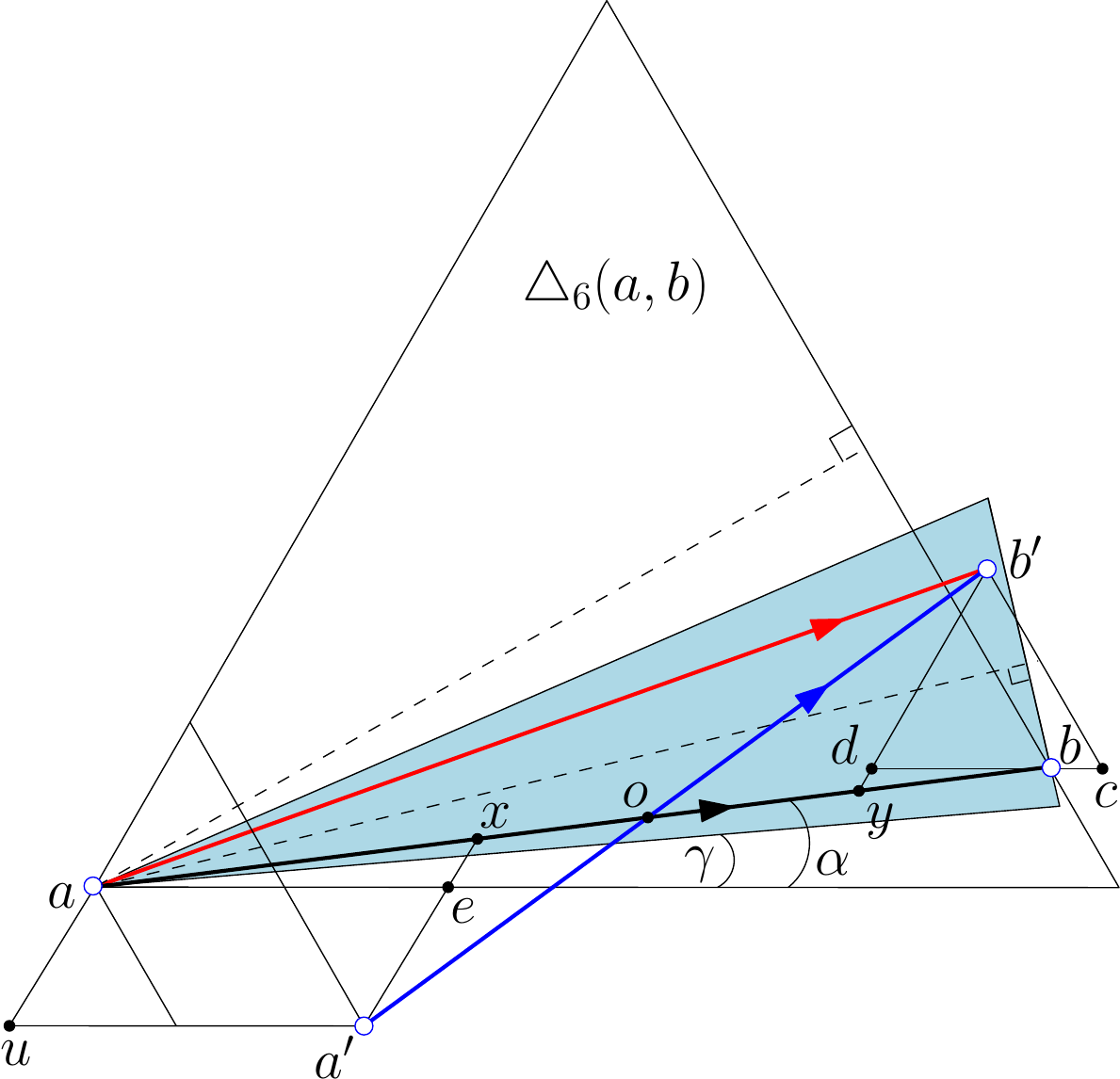} \\
(a) & (b)  
\end{tabular}
\caption{Bounding $|\pp_{\Theta_6}(a, a')|+|\pp_{\Theta_6}(b, b')|$ (a)~\autoref{lem:paa1}: $a'$ above $a$ (b)~\autoref{lem:paasecond}: $a'$ below $a$.} 
\label{fig:casebb}
\end{figure}
%

\begin{restatable}{lemma}{aalemma}
\label{lem:paa1}
Let $a, b, b', a' \in S$ be points in a canonical $\Theta$-configuration, with the additional constraint that $a' \in C_{\{6,2\}}(a)$. Let $\beta$ and $\gamma$ be the angles formed by the horizontal through $a$ with $ab'$ and the lower ray of $C_k(a, b)$, respectively. 
Then 
\begin{equation*}
|\pp_{\Theta_6}(a, a')| + |\pp_{\Theta_6}(b, b')| \le (|ab|+|a'b'|)\cdot T(\gamma) -2 |ab'|\cdot T(\beta) 
\end{equation*}
Here the term $T$ is as defined in~\emph{(\ref{eq:T})}. Furthermore, each edge of $\pp_{\Theta_6}(a,a')$ and $\pp_{\Theta_6}(b,b')$ is strictly smaller than $ab$, for $\theta \le \pi/6$. \emph{[Refer to~\autoref{fig:casebb}a.]}
\end{restatable}



\begin{restatable}{lemma}{aasecondlemma}
\label{lem:paasecond}
Let $a, b, b', a' \in S$ be points in a canonical $\Theta$-configuration, with the additional constraints that $a' \in C_{\{6,6\}}(a)$, and the angle $\alpha$ formed by $ab$ with the horizontal through $a$ is at most $\pi/6$. 
Then 
\begin{equation*}
|\pp_{\Theta_6}(a, a')| + |\pp_{\Theta_6}(b, b')| \le |ab| 
- |a'b'|\cdot\frac{\sin(\pi/3-\alpha-\theta)-\sin\theta}{\sin(\pi/3-\alpha)}
\end{equation*}
Furthermore, each edge of $\pp_{\Theta_6}(a,a')$ and $\pp_{\Theta_6}(b,b')$ is strictly shorter than $ab$, for $\theta \le \pi/12$. \emph{[Refer to~\autoref{fig:casebb}b.]}
\end{restatable}

\begin{restatable}{lemma}{aathirdlemma}
\label{lem:paa5}
Let $a, b, b', a' \in S$ be points in a canonical $\Theta$-configuration, with the additional constraint that either  
$a' \in C_{\{6,5\}}(a)$, or $a' \in C_{\{6,6\}}(a)$ and the angle formed by $ab$ with the horizontal through $a$ is above $\pi/6$. 
Then 
\begin{small}
\begin{equation*}
|\pp_{\Theta_6}(a, a')| + |\pp_{\Theta_6}(b, b')| \le 8|ab|\sin(\theta/2)
\end{equation*}
\end{small}Furthermore, each edge of $\pp_{\Theta_6}(a,a')$ and $\pp_{\Theta_6}(b,b')$ is strictly shorter than $ab$, for $\theta \le \pi/15$. 
\end{restatable}

Our approach to finding a short path in $\Theta\Theta_k$ between the endpoints of each edge in $\Theta_6$ uses induction on the Euclidean lengths of the edges in $\Theta_6$. The following lemma will be useful in proving the inductive step in various situations. 

\begin{lemma}
\label{lem:pab}
Let $a, b, b', a' \in S$ be points in a canonical $\Theta$-configuration, and let $t \ge 1$ be a fixed real value.  Assume that, for each edge $xy \in \Theta_6$ no longer than $ab$, the inequality $|\pp_{\Theta\Theta_k}(x,y)| \le t\cdot|xy|$ holds. 
Let $\pp_{\Theta\Theta_k}(a, b) = \pp_{\Theta\Theta_k}(a, a') \oplus a'b' \oplus \pp_{\Theta\Theta_k}(b', b)$. 
If $|\pp_{\Theta_6}(a, a')| < |ab|$ and $|\pp_{\Theta_6}(b, b')| < |ab|$, then 
\[
|\pp_{\Theta\Theta_k}(a, b)| \le t\cdot |\pp_{\Theta_6}(a, a')| + t\cdot |\pp_{\Theta_6}(b',b)| + |a'b'|
\]
Furthermore, if $|ab|-|\pp_{\Theta_6}(a, a') - |\pp_{\Theta_6}(b',b)| > 0$, then 
$|\pp_{\Theta\Theta_k}(a, b)| \le t\cdot |ab|$  for any real value $t$ such that 
\begin{equation}
\label{eq:t}
t \ge \frac{|ab|/\cos(\theta/2)}{|ab|-|\pp_{\Theta_6}(a, a')| - |\pp_{\Theta_6}(b',b)|}.
\end{equation}
Here the symbol $\oplus$ is used to denote the path concatenation operator. 
\end{lemma}
\begin{proof}
Because $|\pp_{\Theta_6}(a, a')| < |ab|$, each edge on $\pp_{\Theta_6}(a, a')$ must be shorter than $ab$. This along with the lemma statement implies that, for each edge $xy$ on the path $\pp_{\Theta_6}(a, a')$, the inequality $|\pp_{\Theta\Theta_k}(x,y)| \le t\cdot|xy|$ holds. Summing up these inequalities for all edges along the path $\pp_{\Theta_6}(a, a')$ yields $|\pp_{\Theta\Theta_k}(a, a')| \le t\cdot |\pp_{\Theta_6}(a, a')|$.  Similar arguments show that $|\pp_{\Theta\Theta_k}(b, b')| \le t\cdot |\pp_{\Theta_6}(b, b')|$. Thus the 
first inequality stated by this lemma holds. 
Using the upper bound on $|a'b'|$ from~\autoref{lem:abba}, and the assumption that 
$|ab|-|\pp_{\Theta_6}(a, a') - |\pp_{\Theta_6}(b',b)| > 0$, this inequality can be easily reorganized into $|\pp_{\Theta\Theta_k}(a, b)| \le t\cdot|ab|$ for any real value $t$ that satisfies~(\ref{eq:t}).
{\hfill\ABox}\end{proof}

\section{$\Theta\Theta_{6k'}$ is a Spanner, for $k' \ge 4$}
\label{sec:main}
This section presents our main result, which shows that $\Theta\Theta_k$ is a spanner, provided that $k = 6k'$ and $k' \ge 5$ (and so $\theta \le \pi/15)$. In particular, we show that for each edge $ab \in \Theta_6$, there is a path in $\Theta\Theta_{6k'}$ no longer than $8.38|ab|$. This, combined with the result of~\autoref{thm:theta6}, yields our main result that $\Theta\Theta_{6k'}$ is a $16.76$-spanner, for $k' \ge 5$. The spanning ratio decreases to $7.82$ for $k' \ge 6$,  which is superior to the spanning ratio of $11.67$ established in~\cite{jDR12} for $YY_{6k'}$, with $k' \ge 6$. We also show that the spanning ratio of $\Theta\Theta_{6k'}$ drops to $4.64$ for $k' \ge 8$. 

Our approach takes advantage of the fact that each edge $ab \in \Theta_6$ is embedded in an equilateral triangle $\triangle_6(a, b)$ empty of points in $S$. The restriction $k = 6k'$ is necessary in our analysis to guarantee that each cone used in constructing $\Theta_k$ and $\Theta\Theta_k$ is a subset of a cone used in constructing $\Theta_6$, therefore it inherits a large area empty of 
points in $S$. This property is crucial in establishing a ``short'' path in $\Theta\Theta_{k}$ between the endpoints of each edge in $\Theta_6$. Although we search for \emph{undirected} paths in the undirected version of $\Theta\Theta_k$, we sometimes point out the direction of an edge if significant in the context. 

\begin{theorem}
Let $k = 6k'$ be a positive integer, with $k' \ge 5$. For each edge $\arr{ab} \in \Theta_6$, a shortest path in $\Theta\Theta_k$ between $a$ and $b$ satisfies $|\pp_{\Theta\Theta_k}(a, b)| \le t \cdot |ab|$, where $t$ is a positive real with values $8.38$, $3.91$, $2.811$ and $2.32$ corresponding to $k'$ values $5$, $6$, $7$, and above $8$, respectively.
\label{thm:maintheta}
\end{theorem}
\begin{proof}
Recall that $\theta = 2\pi/k$, so in the context of this theorem $\theta \le \pi/15$. Throughout this proof will refer to the value $t$ from the theorem statement as the \emph{stretch} factor, with the understanding that it measures the ``stretch'' in $\Theta\Theta_k$ of an edge $ab \in \Theta_6$, and to be distinguished from the spanning ratio of $\Theta\Theta_k$ (which by~\autoref{thm:theta6} is at most 2$t$). 

The proof is by induction on the Euclidean length of the edges in $\Theta_6$. 
The base case corresponds to a shortest edge $\arr{ab} \in \Theta_6$.
In this case we show that $\arr{ab} \in \Theta_{k}$ and $\arr{ab} \in \Theta\Theta_{k}$. Assume to the contrary that
$\arr{ab} \not\in \Theta_{k}$ and let $\arr{ab'} \in \Theta_{k}$  
be the edge that lies in $C_k(a,b)$. 
~\autoref{lem:paa1} does not impose any restrictions on the relative position of the $b$ and $b'$, therefore the result that 
each edge on $\pp_{\Theta_6}(b, b')$ is strictly shorter than $ab$ applies in this context. 
This contradicts our assumption that $ab$ is a shortest edge in $\Theta_6$. This shows that $\arr{ab} \in \Theta_k$. 
Similar arguments, used in conjunction with Lemmas~\ref{lem:paa1},~\ref{lem:paasecond} and~\ref{lem:paa5} (which distinguish between different locations of $a'$ relative to $a$), show that $\arr{ab} \in \Theta\Theta_{k}$.

Our inductive hypothesis states that the theorem holds for all edges in $\Theta_6$ of length strictly lower than some fixed value $\delta > 0$. To prove the inductive step, pick a shortest edge $\arr{ab} \in \Theta_6$ of length $\delta$ or higher, and find a ``short'' path $\pp_{\Theta\Theta_k}(a, b)$ that satisfies the conditions of the theorem. Let $a'$ and $b'$ be the other two points in $S$ which, along with $a$ and $b$, complete a canonical $\Theta$-configuration: $\arr{ab'} \in \Theta_k$ lies in $C_k(a, b)$, and 
$\arr{a'b'} \in \Theta\Theta_k$ lies in $C_k(b', a)$. Refer to~\autoref{fig:abba}a. Also recall our general assumptions that in a canonical $\Theta$-configuration $ab \in C_{\{6, 1\}}(a)$, and the bisector of $C_k(a,b)$ aligns with, or lies below, the bisector of $\triangle_6(a, b)$. 
The locus of $b'$ is $\triangle_k(a,b) \setminus \triangle_6(a, b)$, which is an area completely inside $C_{\{6, 1\}}(a)$. 
The locus of $a'$ is $\triangle_k(b',a) \setminus \triangle_6(a, b)$, which is an area that may overlap two or three of the cones $C_{\{6, 2\}}(a)$, $C_{\{6, 5\}}(a)$ and $C_{\{6, 6\}}(a)$.  Note that $a'$ may not lie in $C_{\{6, 3\}}(a)$, due to our assumption that the bisector of $\triangle_k(a, b)$ is no higher than the bisector of $\triangle_6(a, b)$. 

Our intent is to use the result of~\autoref{lem:pab} to establish the existence of a path between $a$ and $b$ of length at most $t\cdot|ab|$, for some fixed real constant $t > 1$. The two key ingredients needed by~\autoref{lem:pab} are ``short'' paths in $\Theta_6$ between $a$ and $a'$, and between $b$ and $b'$. 
We discuss three cases,  depending on whether $a'$ lies in $C_{\{6,2\}}(a)$, $C_{\{6,5\}}(a)$ or $C_{\{6,6\}}(a)$. The case 
$a' \in C_{\{6,5\}}(a)$ is the simplest, so we will save it for last. 
Let $\alpha$, $\beta$ and $\gamma$ be the angles formed by the horizontal through $a$ with $ab$, $ab'$, and the lower ray of $C_k(a, b)$, respectively. 

\paragraph{Case $a' \in C_{\{6,2\}}(a)$.} This case is depicted in~\autoref{fig:casebb}a. By \autoref{lem:paa1}, we have 
\begin{equation}
\label{eq:paa1bound}
 |\pp_{\Theta_6}(a, a')| + |\pp_{\Theta_6}(b, b')| \le (|ab|+|a'b'|)\cdot T(\gamma) - 2|ab'|\cdot T(\beta) 
\end{equation}
where $T$ is as defined in~(\ref{eq:T}). 
Notice the restrictions on the angles $\beta$ and $\gamma$:
\begin{eqnarray}
\nonumber 0  & \le \gamma \le  & \pi/6-\theta/2 \\
\gamma & \le \beta \le &\gamma+\theta 
\label{eq:gamma}
\end{eqnarray}
The upper bound on $\gamma$ is due to our assumption that the bisector of $\triangle_k(a,b)$ is no higher than the bisector of $\triangle_6(a,b)$. The bounds on $\beta$ follow immediately from the definitions of $\gamma$ and $\beta$. 
Next we determine a maximum for the quantity on the right hand side of~(\ref{eq:paa1bound}).  
We consider two situations, depending on ranges of $\beta$, which affect the sign of $T(\beta)$. Observe that $T(\gamma)$ is always positive, since $\pi/3-\gamma > \gamma$ for any $\gamma < \pi/6$. 

Assume first that $\beta \le \pi/6$, so $ab'$ is no higher than the bisector of $\triangle_6(a, b)$. 
In this case $\beta \le \pi/3-\beta$ and $\sin\beta \le \sin(\pi/3-\beta)$, therefore $T(\beta)$ is positive. Substituting in~(\ref{eq:paa1bound}) the upper bound on $|a'b'|$ and the lower bound on $|ab'|$ from~\autoref{lem:abba} yields
\begin{equation*}
\frac{ |\pp_{\Theta_6}(a, a')| +  |\pp_{\Theta_6}(b, b')|}{|ab|} \le T(\gamma) + \frac{T(\gamma)}{\cos(\theta/2)} - 2T(\beta)\cdot\frac{\sin(\pi/3+\gamma)}{\sin(\pi/3+\beta)}
 \label{eq:paa1bound1}
\end{equation*}
Let $X(\theta, \gamma, \beta)$ denote the quantity on the right hand side of the inequality above. 
Note that $X(\theta, \gamma, \beta)$ increases as $\theta$ increases, therefore $X(\theta, \gamma, \beta) \le X(\pi/15, \gamma, \beta)$ for $\theta \le \pi/15$. ~\autoref{fig:paa1plot}a shows how  $X(\theta, \gamma, \beta)$ varies with $\gamma \in [0, \pi/6-\theta/2]$ and $\beta \in (\gamma, \min\{\pi/6, \gamma+\theta\}]$, for fixed $\theta = \pi/15$. 
%
\begin{figure}[hptb]
\centering
\begin{tabular}{c@{\hspace{0.1\linewidth}}c}
\includegraphics[width=0.4\linewidth]{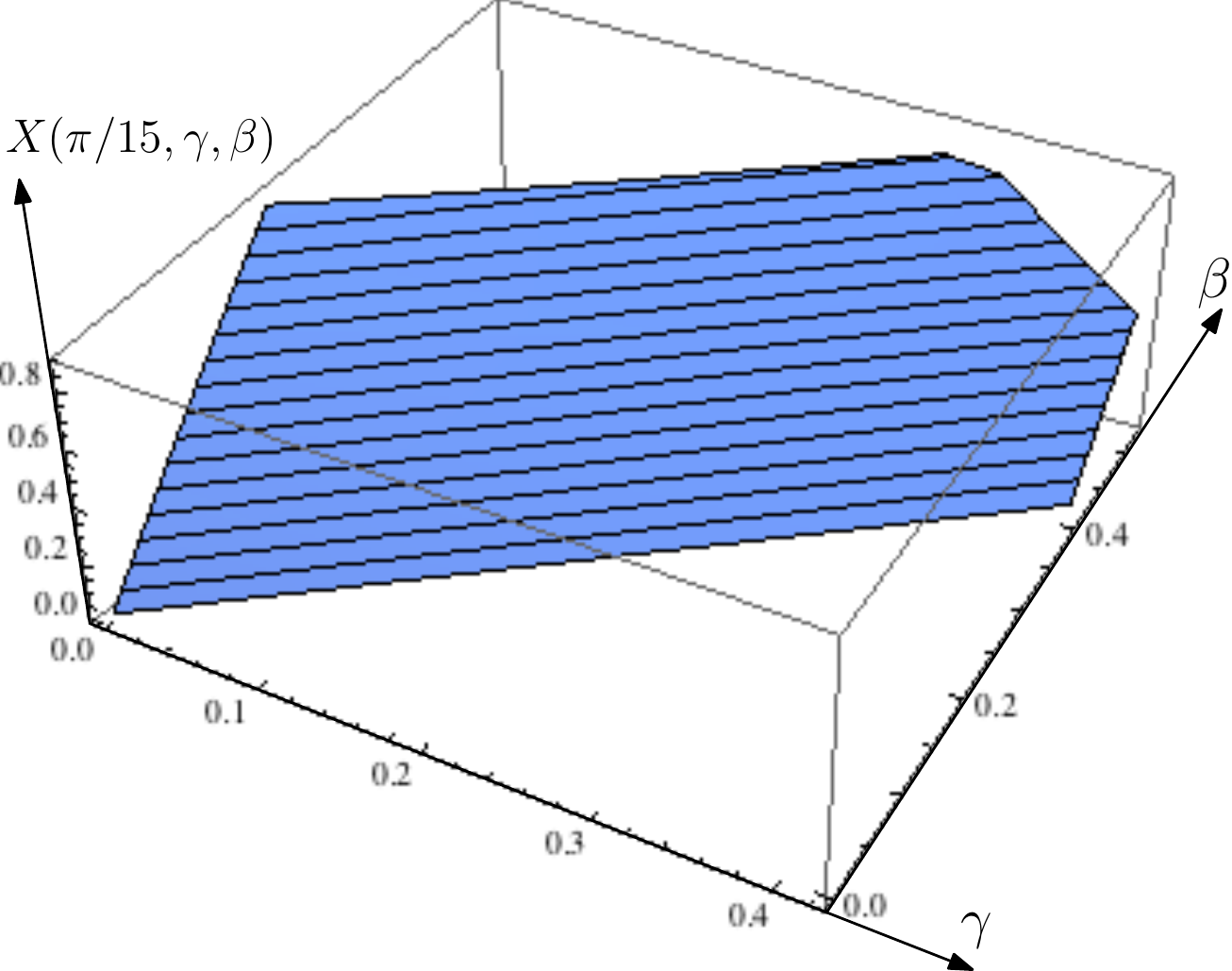} &
\includegraphics[width=0.4\linewidth]{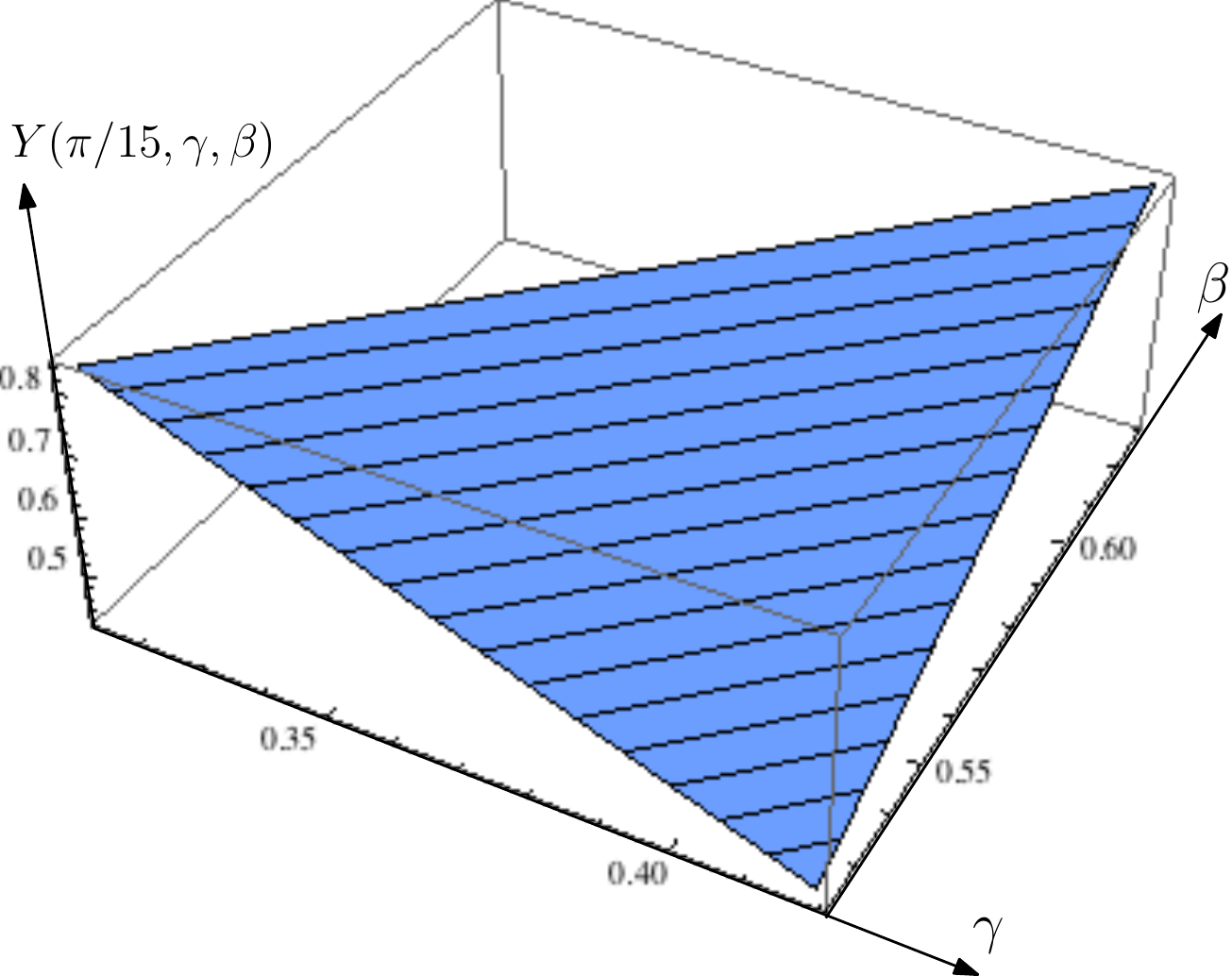} \\
(a) & (b) 
\end{tabular}
\caption{Case $a' \in C_{\{6,2\}}(a)$: upper bound on $|\pp_{\Theta_6}(a, a')|+ |\pp_{\Theta_6}(b, b')|$ for $\theta = \pi/15$ and $\gamma \in [0, \pi/6-\theta/2]$ (a) $\alpha \in (\gamma, \min\{\pi/6, \gamma+\theta\}]$ (b) $\alpha \in (\pi/6, \min\{\pi/6+\theta/2, \gamma+\theta\}]$}
\label{fig:paa1plot}
\end{figure}
%
It can be verified that 
$|\pp_{\Theta_6}(a, a')|+  |\pp_{\Theta_6}(b, b')| < 0.88 |ab|$, for any $0 < \theta \le \pi/15$. This along with~\autoref{lem:pab} yields a stretch factor $t = 8.3760$ for the path in $\Theta\Theta_{k}$ between $a$ and $b$. 
The stretch factor $t$ decreases with $\theta$ as shown in the second column of Table~\ref{tab:paa1bound}. 

\begin{table}[hptb]
\begin{center}
\begin{tabular} {|c|c|c|}
\hline
\multirow{3}{*}{$\theta$} & \multicolumn{2}{|c|}{Case $a' \in C_{\{6,2\}}(a)$: stretch factor $t$ from~\autoref{lem:pab}} \\
\cline{2-3}
& ~~~~~~~~~$0 < \beta \le \pi/6$~~~~~~~~ & $\pi/6 < \beta \le \pi/6+\theta/2$ \\
\hline
$\pi/15$ & $8.3760$ & $6.2720$ \\
$\pi/18$ & $3.9058$ & $3.3377$ \\
$\pi/21$ & $2.8109$ & $2.5014$ \\
$\pi/24$ & $2.3159$ & $2.1057$ \\
\hline
\end{tabular}
\end{center}
\caption{Case $a' \in C_{\{6,2\}}(a)$: 
real constant $t$ from~\autoref{lem:pab} for various $\theta$ values.}
\label{tab:paa1bound}
\end{table}

Assume now that $\pi/6 < \beta \le \pi/6+\theta/2$, so $ab'$ lies above the bisector of $\triangle_k(a, b)$. 
In this case $T(\beta)$ is negative, and by~(\ref{eq:gamma}) we have $\gamma \ge \pi/6-\theta$. Substituting in~(\ref{eq:paa1bound}) the upper bound on $|ab'|$ and $|a'b'|$ from~\autoref{lem:abba} yields
\begin{equation}
 \frac{ |\pp_{\Theta_6}(a, a')| +  |\pp_{\Theta_6}(b, b')|}{|ab|} \le  T(\gamma) + \frac{T(\gamma)-2T(\beta)}{\cos(\theta/2)}
\label{eq:paa1bound2}
\end{equation}
Let $Y(\theta, \gamma, \beta)$ denote the quantity on the right hand side of 
the inequality above.  Because $T(\gamma)$ is positive and $T(\beta)$ is negative, $T(\gamma)-2T(\beta)$ is positive and therefore  $Y(\theta, \gamma, \beta)$ increases as $\theta$ increases. It follows that $Y(\theta, \gamma, \beta) \le Y(\pi/15, \gamma, \beta)$ for $\theta \le \pi/15$. ~\autoref{fig:paa1plot}b shows how $Y(\theta, \gamma, \beta)$ varies with $\gamma \in [\pi/6-\theta, \pi/6-\theta/2]$ and $\beta \in (\pi/6, \gamma+\theta]$, for $\theta = \pi/15$. It can be verified that 
$|\pp_{\Theta_6}(a, a')|+  |\pp_{\Theta_6}(b, b')|  < 0.8397 |ab|$, for any $0 < \theta \le \pi/15$.
This along with~\autoref{lem:pab} yields a stretch factor $t = 6.2720$ for the path in $\Theta\Theta_{k}$ between $a$ and $b$. 
The stretch factor $t$ decreases with $\theta$ as shown in the third column of Table~\ref{tab:paa1bound}. 


\paragraph{Case $a' \in C_{\{6,6\}}$.} This case is depicted in~\autoref{fig:casebb}b. We discuss two situations, depending on whether $ab$ lies above or below the bisector of $\triangle_6(a, b)$. Assume first that $ab$ is no higher than the bisector of $\triangle_6(a, b)$, so $\alpha \le \pi/6$. Thus we are in the context of~\autoref{lem:paasecond}, which gives us an upper bound  
$|\pp_{\Theta_6}(a, a')| + |\pp_{\Theta_6}(b, b')|  \le |ab| - |a'b'|\cdot Z(\theta, \alpha)$, where 
\begin{equation*}
Z(\theta,\alpha)  = \frac{\sin(\pi/3-\theta-\alpha)-\sin\theta}{\sin(\pi/3-\alpha)}
\end{equation*}
Note that $Z(\theta,\alpha)$ decreases as $\theta$ increases, therefore $Z(\theta,\alpha) \ge Z(\pi/15, \alpha)$ for any $\theta \le \pi/15$. It can be verified that $Z(\theta,\alpha) \ge m = 0.2022$, for any $\theta \le \pi/15$. 
By~\autoref{lem:pab}, we have 
\[
\pp_{\Theta\Theta_k}(a, b) \le t|ab|-t|a'b'|\cdot Z(\theta,\alpha) +|a'b'|
\] 
Simple calculations show that the right hand side of the inequality above does not exceed $t|ab|$ for any $t \ge 4.945 \ge 1/m$. This bound decreases with $\theta$ as shown in the second column of Table~\ref{tab:paa2bound}. 

Assume now that $ab$ lies above the bisector of $\triangle_6(a, b)$, so $\alpha > \pi/6$. Intuitively, this forces $a$ and $a'$ to lie close to each other (for sufficiently small $\theta$ values), and similarly for $b$ and $b'$, so we can work with somewhat looser upper bounds without exceeding the spanning ratio established so far. Our context matches the context of~\autoref{lem:paa5}, which tells us that $|\pp_{\Theta_6}(a, a')| + |\pp_{\Theta_6}(b, b')| \le X(\theta) = 8|ab|\sin(\theta/2)$. The bound $X(\theta)$ increases with $\theta$, therefore $X(\theta) \le X(\pi/15) \le 0.8363$.  This together with~\autoref{lem:pab} yields 
$|\pp_{\Theta\Theta_k}(a, b)|  < t \cdot |ab|$ for any $t \ge 6.1397$. 
This bound decreases with $\theta$ as shown in the third column of Table~\ref{tab:paa2bound}. 


\begin{table}[hptb]
\begin{center}
\begin{tabular} {|c|c|c|}
\hline
\multirow{3}{*}{$\theta$} & \multicolumn{2}{|c|}{Case $a' \in C_{\{6,6\}}(a)$: stretch factor $t$ from~\autoref{lem:pab}} \\
\cline{2-3}
& ~~~~~~~~$0 \le \alpha \le \pi/6$~~~~~~~~~~~ & $\pi/6 < \alpha \le \pi/6+\theta/2$ \\
\hline
$\pi/15$ & $4.9454$ & $6.1397$ \\
$\pi/18$ & $2.9697$ & $3.3157$ \\
$\pi/21$ & $2.3117$ & $2.4936$ \\
$\pi/24$ & $1.9829$ & $2.1020$ \\
\hline
\end{tabular}
\end{center}
\caption{Case $a' \in C_{\{6,6\}}(a)$: real constant $t$ from~\autoref{lem:pab} for various $\theta$ values.}
\label{tab:paa2bound}
\end{table}
\paragraph{Case $a' \in C_{\{6,5\}}$.} The bound on $|\pp_{\Theta_6}(a, a')| + |\pp_{\Theta_6}(b, b')|$ provided by~\autoref{lem:paa5} applies here as well, therefore the analysis for this case is identical to the one for the previous case (with $b' \in C_{\{6,6\}}$ and $ab$ above the bisector of $\triangle_6(a, b)$), yielding the spanning ratios listed in the third column of Table~\ref{tab:paa2bound}. 

To derive the results listed in Tables~\ref{tab:paa1bound} and~\ref{tab:paa2bound}, we worked with a quadruplet of \emph{distinct} points $a,b,a',b'$ in a  $\Theta$-configuration. The cases where $a$ and $a'$ coincide, or $b$ and $b'$ coincide, are special instances of this general case and yield lower stretch factors. 
The results listed in Tables~\ref{tab:paa1bound} and~\ref{tab:paa2bound} 
indicate that the stretch factor is highest when $a'$ lies above $a$ and $ab'$ is below the bisector of $\triangle_6(a,b)$. The largest stretch factor value is $t = 8.376$ for $\theta = \pi/15$, and it drops to $3.91$, $2.82$ and $2.32$ for $\theta$ values 
$\pi/18$, $\pi/21$ and $\pi/24$, respectively. This completes the proof.
{\hfill\ABox}\end{proof}
Combined with the result of~\autoref{thm:theta6}, the result of~\autoref{thm:maintheta} yields the main result of this paper, stated by~\autoref{thm:main} below. 

\begin{theorem} 
\label{thm:main}
The graph $\Theta\Theta_k$, with $k = 6k'$ and $k' \ge 5$, is a $16.76$-spanner. The spanning ratio decreases to $7.82$, $5.63$ and $4.64$ as $k'$ increases to $6$, $7$, and above $8$, respectively.
\end{theorem}

\section{Conclusions}
In this paper we present the first results on the spanning property of $\Theta\Theta_k$-graphs. We show that, for any integer $k' \ge 5$, the graph $\Theta\Theta_{6k'}$ is a spanner with spanning ratio $16.76$. The spanning ratio drops to $7.82$ for $k' \ge 6$, which is superior to the spanning ratio of $11.67$ established in~\cite{jDR12} for $YY_{6k'}$, with $k' \ge 6$. The framework of our analysis seems inadequate to handle all graphs $\Theta\Theta_k$, for all $k>6$, because it relies on the fact that each cone used in constructing $\Theta\Theta_k$ is a subset of a cone used in constructing $\Theta_6$. 
It is unclear whether a fundamentally new technique is  required to handle all $\Theta\Theta_k$ graphs, for $k \ge 6$. Proving or disproving that these graphs are spanners remains the main open problem in this area. 

\bibliographystyle{plain}
\bibliography{spannerbib}

\newpage
\section*{Appendix: Deferred Proofs}
\label{sec:defer}


\subsection{Proof of~\autoref{lem:abba}}

\abbalemma*
\begin{proof}
Let $h$ be the height of the isosceles triangle $\triangle_k(a, b)$, and let $s$ be the length of its two equal sides. They are related by $h = s\cos(\theta/2)$. Because both $ab$ and $ab'$ lie inside $\triangle_k(a, b)$, their length may not exceed $s$. Also $|ab|$ may not be lower than $h$, since $b$ is on the base of $\triangle_k(a, b)$. This implies 
$|ab| \ge h = s\cos\theta/2 \ge |ab'|\cos\theta/2$, so the upper bound on $|ab'|$ holds.  
Observe now that $\triangle_k(a, b)$ and $\triangle_k(b', a)$ are similar, and the side length of $\triangle_k(b', a)$  does not exceed $s$ (because $b'$ lies inside $\triangle_k(a,b)$). Similar arguments can then be used to establish the same upper bound on $|a'b'|$. 

Let $i$ be the intersection point between $ab'$ and the right side of 
$\triangle_6(a, b)$. Then $|ab'| \ge |ai|$. By the Law of Sines applied on $\triangle abi$, we have 
$|ab|/\sin\ang{aib} = |ai|/\sin\ang{abi}$. Let $r$ be the lower right corner of $\triangle_6(a, b)$. Note that $\ang{aib} = \pi/3-\beta$ (as angle interior to $\triangle air$) and $\ang{abi} = \pi/3+\ang{bar} \ge \pi/3+\gamma$ (as angle exterior to $\triangle abr$). Also because $\beta \le \pi/6$, $\ang{abi}$ is acute, therefore $\sin\ang{abi} \ge \sin(\pi/3+\gamma)$. These together show that $|ai| \ge |ab|\sin(\pi/3+\gamma)/\sin(\pi/3+\alpha)$, so the lower bound on $|ab'|$ holds. This completes the proof. 
%
{\hfill\ABox}\end{proof}

\subsection{Proof of~\autoref{lem:paa1}}
\aalemma*
\begin{proof}
First we determine an upper bound on $|\pp_{\Theta_6}(b, b')|$. Let $c$ and $d$ be the right and left corners of $\triangle_6(b', b)$, respectively. Because $b'$ is interior to $\triangle_k(a, b)$, the perpendicular from $b'$ to the bisector of $\triangle_k(a, b)$ intersects the line segment $ab$, so  the perpendicular from $b'$ to $cd$ falls left of $b$. This implies that $|bc| \le |bd|$. 
Note that $\triangle_6(b',b)$ meets the conditions of~\autoref{lem:thetapath}, with $\triangle_6(b,d)$ empty of points in $S$, therefore   
$|\pp_{\Theta_6}(b, b')| \le |bc| + |b'c| \le |bd| + |b'd|$. 
Let the horizontal through $a$ intersect the left rays of $C_{\{6,5\}}(b)$ and $C_{\{6,5\}}(b')$ in points $x$ and $x'$, respectively. 
Then $|bd| = |ax|-|ax'|$ and $|b'd| = |b'x'|-|bx|$, so we have 
\begin{equation}
|\pp_{\Theta_6}(b, b')| \le (|ax|-|ax'|)+(|b'x'|-|bx|)
\label{eq:pbb1}
\end{equation}
We determine $|ax|$ and $|bx|$ in terms of $|ab|$ by applying the Law of Sines on $\triangle abx$: 
$|ax|/\sin\ang{abx}=|bx|/\sin\ang{bax}=|ab|/\sin\ang{axb}$. Note that $\ang{axb} = 2\pi/3$,  therefore both $\ang{bax}$ and 
$\ang{abx}$ are acute.  This along with the fact that $\ang{bax} \ge \gamma$ implies $\sin\ang{bax} \ge \sin\gamma$, and 
 $\ang{abx} = \pi/3-\ang{bax} \le \pi/3-\gamma$ implies $\sin\ang{abx} \le \sin(\pi/3-\gamma)$. Combining these inequalities together yields
\begin{equation}
|ax| \le |ab|\cdot\frac{\sin(\pi/3-\gamma)}{\sin(\pi/3)} \mbox{~~and~~} |bx| \ge |ab|\cdot\frac{\sin\gamma}{\sin(\pi/3)}
\label{eq:pbb2}
\end{equation}
Next we determine $|b'x'|$ and $|ax'|$ in terms of $|ab'|$ by applying the Law of Sines on $\triangle ab'x'$: 
$|ax'|/\sin\ang{ab'x'}=|b'x'|/\sin\ang{b'ax'}=|ab'|/\sin(\pi/3)$. Plugging in the angle values 
$\ang{b'ax'} = \beta$ and $\ang{ab'x'} = \pi/3-\beta$ yields 
\begin{equation}
|ax'| = |ab'|\cdot\frac{\sin(\pi/3-\beta)}{\sin(\pi/3)} \mbox{~~and~~} |b'x'| = |ab'|\cdot\frac{\sin\beta}{\sin(\pi/3)}
\label{eq:pbb3}
\end{equation}
Combining inequalities~(\ref{eq:pbb1}),~(\ref{eq:pbb2}) and~(\ref{eq:pbb3}) together yields 
\begin{equation}
|\pp_{\Theta_6}(b, b')| \le |ab|\cdot T(\gamma) - |ab'|\cdot T(\beta)
\label{eq:pbbfinal}
\end{equation}
Next we determine an upper bound on $|\pp_{\Theta_6}(a, a')|$.
Let $u$ be the left corner of $\triangle_6(a', a)$ (refer to~\autoref{fig:casebb}a.) By~\autoref{lem:thetapath}, $|\pp_{\Theta_6}(a', a)| \le |au| + |a'u|$.  
Let the horizontal through $b'$ intersect the left side of $\triangle_6(a, b)$ and the line supporting $a'u$ in points $y$ and $y'$, respectively. Then $|au| = |b'y'|-|b'y|$ and $|a'u| = |ay|-|a'y'|$. These together imply  
\begin{equation}
|\pp_{\Theta_6}(a, a')| \le (|ay|-|a'y'|)+(|b'y'|-|b'y|)
\label{eq:paa1}
\end{equation}
%
Note that $|ay| = |b'x'|$ and $|b'y| = |ax'|$, so the bounds from~(\ref{eq:pbb3}) apply here as well. 
Next we determine $|a'y'|$ and $|b'y'|$ in terms of $|a'b'|$ by applying the Law of Sines on $\triangle a'b'y'$: 
$|a'y'|/\sin\ang{a'b'y'}=|b'y'|/\sin\ang{y'a'b'}=|a'b'|/\sin\ang{a'y'b'}$. Because the upper ray of 
$C_k(b',a')$ is parallel to the lower ray of $C_k(a, b)$,  we have 
$\ang{a'b'y'} \ge \gamma$ and $\ang{y'a'b'} = \pi/3-\ang{a'b'y'} \le \pi/3- \gamma$. 
Since both angles are acute, we get  
$\sin\ang{a'b'y'} \ge \sin\gamma$ and $\sin\ang{y'a'b'} \le \sin(\pi/3-\gamma)$. These together imply 
\begin{equation}
|a'y'| \ge |a'b'|\cdot\frac{\sin\gamma}{\sin(\pi/3)} \mbox{~~and~~} |b'y'| \le |a'b'|\cdot\frac{\sin(\pi/3-\gamma)}{\sin(\pi/3)}
\label{eq:paa2}
\end{equation}
Combining inequalities~(\ref{eq:paa1}),~(\ref{eq:pbb3}) and~(\ref{eq:paa2}) together yields 
\begin{equation*}
|\pp_{\Theta_6}(a, a')| \le |a'b'|\cdot T(\gamma) - |ab'|\cdot T(\beta)
\label{eq:paafinal}
\end{equation*}
This along with~(\ref{eq:pbbfinal}) settles the first part of the lemma. 
We now turn to the second claim of the lemma. By~\autoref{lem:thetapath}, each edge on 
$\pp_{\Theta_6}(b,b')$ is no longer than $|b'd| \le |b'x'| = |ab'|\sin\beta/\sin(\pi/3)$ (cf.~(\ref{eq:pbb3})), and  
each edge on $\pp_{\Theta_6}(a,a')$ is no longer than $|a'u| \le |ay| = |b'x'|$. 
To simplify discussion, let $A = |ab'|\sin\beta/\sin(\pi/3)$. 
It suffices to show that $A < |ab|$ in order to settle the second part of the lemma. 
Because the bisector of $\triangle_k(a, b)$ is no higher than the bisector of $\triangle_6(a, b)$, we have that $\beta \le \pi/6+\theta/2$, therefore $A < |ab'|\sin(\pi/6+\theta/2)/\sin(\pi/3)$. Substituting the upper bound on $|ab'|$ from~\autoref{lem:abba} yields 
\[
A < |ab|\frac{\sin(\pi/6+\theta/2)}{\sin(\pi/3)\cos(\theta/2)}
\]
 It can be verified that the right hand side of this inequality is strictly smaller than $|ab|$, for any $\theta \le \pi/6$. This completes the proof.
{\hfill\ABox}\end{proof}

\subsection{Proof of~\autoref{lem:paasecond}}

\aasecondlemma*

\begin{proof}
We define the following points: $c$ and $d$ are the right and left corners of $\triangle_6(b', b)$; $u$ is the left corner of $\triangle_6(a', a)$; $x$ and $e$ are the points where the right ray of $C_{\{6,2\}}(a')$ intersects $ab$ and the horizontal through $a$, respectively; $y$ is the point where the line supporting $b'd$ intersects $ab$; and $o$ is the intersection point between $ab$ and $a'b'$. Refer to~\autoref{fig:casebb}b. Arguments similar to the ones used in the proof of~\autoref{lem:paa1} show that $|\pp_{\Theta_6}(b, b')| \le |b'd| + |bd|$. This along with $|b'd| \le |b'y|$ and $|bd| \le |by|$ implies 
\begin{equation*}
|\pp_{\Theta_6}(b, b')| \le |by| + |b'y|
\label{eq:paa51}
\end{equation*}
By~\autoref{lem:thetapath} we have  
$|\pp_{\Theta_6}(a, a')| \le |a'u| + |au| = |ae| + |a'e| \le |ax| + |a'x|$. This together with the inequality above and the fact that 
$|by| + |ax| = |ab| - |xy|$,  yields 
\begin{equation}
|\pp_{\Theta_6}(a, a')| + |\pp_{\Theta_6}(b, b')| \le |ab| - |xy| + (|a'x|+|b'y|)
\label{eq:paa52}
\end{equation}
Using the similarity property of $\triangle a'ox$ and $\triangle b'oy$, we derive $|xy| =  |a'b'| \cdot  |xo| / |a'o|$ and 
$|a'x|+|b'y| = |a'b'|\cdot |a'x|/|a'o|$.
Using the Law of Sines on $\triangle a'ox$, we derive 
$|xo|/|a'o| = \sin\ang{xa'o}/\sin\ang{a'xo}$ and 
$|a'x|/|a'o| = \sin\ang{a'ox}/\sin\ang{a'xo}$.
Observe that $\ang{a'ox} \le \theta$ (because the ray shooting from $b'$ towards $a$, parallel to $ab$, lies inside $C_k(b',a')$ of angle $\theta$, and $\ang{a'ox}$ is equal to the angle formed by this ray with $a'b'$), and $\ang{a'xo} =2\pi/3+\alpha$ (as angle exterior to $\triangle aex$). It follows that $\ang{xa'o} > \pi/3-\alpha-\theta$. These together imply
\begin{equation*}
|xy| > |a'b'|\cdot\frac{\sin(\pi/3-\alpha-\theta)}{\sin(\pi/3-\alpha)} \mbox{~~~and~~~}  |a'x|+|b'y| < |a'b'|\cdot\frac{\sin\theta}{\sin(\pi/3-\alpha)}
\end{equation*}
These inequalities along with~(\ref{eq:paa52}) yield
the upper bound on $|\pp_{\Theta_6}(a, a')|+|\pp_{\Theta_6}(b, b')|$ stated by this lemma.


For the second part of the lemma, it can be verified that the term $(\sin(\pi/3-\alpha-\theta)-\sin\theta)/\sin(\pi/3-\alpha)$ is strictly positive for any $t \in (0, \pi/12]$ and $\alpha \in [0, \pi/3]$. This along with the upper bound established by this lemma shows that $|\pp_{\Theta_6}(a, a')| + |\pp_{\Theta_6}(b, b')| < |ab|$, therefore each edge on each of the paths $\pp_{\Theta_6}(a,a')$ and $\pp_{\Theta_6}(b,b')$ is strictly shorter than $ab$. This completes the proof.
{\hfill\ABox}\end{proof}

\subsection{Proof of~\autoref{lem:paa5}}

\aathirdlemma*

\begin{proof}
The conditions stated by the lemma suggest that either $a$ and $a'$ lie close to each other (if $a' \in C_{\{6,5\}}(a)$), or $b$ and $b'$ lie close to each other (if $ab$ is above the bisector of $\triangle_6(a, b)$). Intuitively, the upper bounds established for these two cases must be within a small factor of each other.  

Let $o$ be the intersection point between $ab$ and $a'b'$. By the lemma statement $b'$ lies below the horizontal through $a$, therefore the point $o$ exists.  
Observe that a ray shooting from $b'$ towards $a$, parallel to $ab$, lies inside $C_k(b',a')$ of angle $\theta$, and $\ang{aoa'}$ is equal to the angle formed by this ray with $a'b'$, therefore $\ang{aoa'} \le \theta$. By the Law of Sines applied on triangle $\triangle aoa'$, we have 
$|aa'|/\sin\ang{aoa'} = |oa'|/\sin\ang{a'ao}$. This along with~\autoref{thm:theta6} and the fact that $\ang{aoa'} \le \theta$ implies 
\begin{equation}
|\pp_{\Theta_6}(a, a')| \le 2|oa'| \cdot \frac{\sin\theta}{\sin\ang{a'ao}}
\label{eq:small1}
\end{equation}
Similarly arguments used on $\triangle bob'$ show that  
\begin{equation}
|\pp_{\Theta_6}(b, b')| \le 2|ob'| \cdot \frac{\sin\theta}{\sin\ang{b'bo}}
\label{eq:small2}
\end{equation}
Consider first the case where $a' \in C_{\{6,6\}}(a)$, and $ab$ is above the bisector of $\triangle_6(a, b)$. In this case $\ang{b'bo} > \pi/2$ and $\ang{a'ao} > \pi/6$ (since $a'$ is below the horizontal through $a$). By the definition of a $\Theta$-configuration, the bisector of $C_k(a, b)$ lies below the bisector of $\triangle_6(a, b)$, therefore the angle formed by $ab$ with the bisector of $\triangle_6(a, b)$ is at most $\theta/2$. It follows that 
$\ang{a'ao} < \pi/2+\theta/2 < \pi/2 + \pi/6$ and similarly $\ang{b'bo} < \pi/2+\pi/6$. These together show that 
$\sin\ang{a'ao} > \sin(\pi/6)$ and $\sin\ang{b'bo} > \sin(\pi/6)$, which along with~(\ref{eq:small1}) and~(\ref{eq:small2}) yield 
\begin{equation}
\label{eq:small3}
|\pp_{\Theta_6}(a, a')| +|\pp_{\Theta_6}(b, b')| \le 2|a'b'|\cdot \frac{\sin\theta}{\sin(\pi/6)} = 4|a'b'|\sin\theta
\end{equation}
Substituting  the upper bound on $|a'b'|$ from~\autoref{lem:abba} results in $|\pp_{\Theta_6}(a, a')| +|\pp_{\Theta_6}(b, b')| \le 4|ab| \sin\theta/\cos(\theta/2) = 8|ab|\sin(\theta/2)$. Thus the upper bound claimed by the lemma holds for this case. 

Assume now that $b' \in C_{\{6,5\}}(a)$. In this case $\ang{a'ao} \ge \pi/3$, and similarly $\ang{b'bo} \ge \pi/3$ (because $b'$ lies exterior to $\triangle_6(a,b)$ and above $b$). Since neither of these angles can extend as far as $\pi/2+\pi/3$, the inequalities 
$\sin\ang{a'ao} \ge \sin(\pi/3)$ and $\sin\ang{b'bo} \ge \sin(\pi/3)$ hold. These along with~(\ref{eq:small1}) and~(\ref{eq:small2}) yield 
\begin{equation*}
|\pp_{\Theta_6}(a, a')| +|\pp_{\Theta_6}(b, b')| \le 2|a'b'|\cdot \frac{\sin\theta}{\sin(\pi/3)} < 2|a'b'|\cdot \frac{\sin\theta}{\sin(\pi/6)}
\label{eq:small4}
\end{equation*}
This shows that the bound from~(\ref{eq:small3}) established for the previous case applies in this case as well. This settles the first part of the lemma. For the second part, simple calculations show that $8\sin(\theta/2) < 1$ for any $\theta <= \pi/15$. This implies that $|\pp_{\Theta_6}(a, a')| + |\pp_{\Theta_6}(b, b')| < |ab|$, therefore each edge of $\pp_{\Theta_6}(a,a')$ and $\pp_{\Theta_6}(b,b')$ is strictly shorter than $ab$. This completes the proof.
{\hfill\ABox}\end{proof}
\end{document}